\documentclass[runningheads,a4paper]{llncs}

\usepackage{amssymb}
\usepackage{graphicx}

\usepackage{url}
\urldef{\mailsa}\path| mahyunst@yahoo.com, mahyuddin@usu.ac.id |
\newcommand{\keywords}[1]{\par\addvspace\baselineskip
\noindent\keywordname\enspace\ignorespaces#1}
\newcommand{\kotak}{\rule{.08in}{.08in}}
\newcommand{\kotakt}{\rule{.08in}{.12in}}
\begin{document}

\mainmatter  
\title{Simple Search Engine Model: Adaptive Properties}

\titlerunning{Simple Search Engine Model: Adaptive Properties}
\author{Mahyuddin K. M. Nasution
\authorrunning{Mahyuddin K. M. Nasution}
\institute{Information Technology Department, \\ Fakultas Ilmu Komputer dan Teknologi Informasi\\
Universitas Sumatera Utara, Padang Bulan, Medan 20155, Sumatera Utara, Indonesia\\
\mailsa\\}}
\toctitle{}
\tocauthor{}
\maketitle

\begin{abstract}
In this paper we study the relationship between query and search engine by exploring the adaptive properties based on a simple search engine. We used set theory and utilized the words and terms for defining singleton space of event in a search engine model, and then provided the inclusion between one singleton to another.
\keywords{singleton space, information retrieval, search term, query.}
\end{abstract}

\section{Introduction}

A search engine on the World Wide Web, in brief we called it as Web, is extensively important to help users to find relevant information. The search engines have some features for servicing the tasks and subtasks that directly or indirectly uses the techniques such as indexing, filters, hub, page rank, hits, and etc \cite{croft2010}, but to access any information in Web the users need the formulating a query about the required information. In this case, the query has become the leading paradigm to find the information, whereby the information retreival (IR) is concerned with answering information need as accurately as possible. However, the users lack understand a formulae of query. Moreover, almost all of search of engines is not provide any function to find the special cases such as entity or actors. Therefore, the major challenge in information access is to provide the riched and trusted information. This paper is aimed at generating some adaptive properties of relation between an search engine and a query.

\section{Basic Concept and Motivation}
Let objects (entities or attributes) can be given literally, like the literal text of "Social Network", then all meaning of objects based on words is represented by the literal objects itself. To realize it, first we define formally that a word $w$ is the basic unit of discrete data, defined to be an item from a vocabulary indexed by $\{1,\dots,K\}$, where $w_k = 1$ if $k\in K$, and $w_k=0$ otherwise \cite{blei2003}. Then, we define some instances related to words.
  
\begin{definition}
\label{def:term}
A term $t_x$ consists of at least one or a set of words in a pattern, or $t_k = (w_1w_2\dots w_l)$, $l\leq k$, $k$ is a number of parameters representing word $w$, $l$ is the number of tokens (vocabularies) in $t_k$, $|t_k| = k$ is size of $t_k$. \kotak
\end{definition}

We define a simple search engine as follows.

\begin{definition}
\label{def:searchengine}
Let a set of web pages indexed by search engine be $\Omega$, i.e., a set contains ordered pair of the terms $t_{k_i}$ and the web pages $\omega_{k_j}$, or $(t_{k_i},w_{k_j})$, $i=1,\dots,I$, $j = 1,\dots,J$. The relation table that consists of two columns $t_k$ and $\omega_k$ is a representation of $(t_{k_i},\omega_{k_j})$ where $\Omega_k = \{(t_k,\omega_k)_{ij}\}\subset\Omega$ or $\Omega_k = \{\omega_{k_1},\dots,\omega_{k_j}\}$. The cardinality of $\Omega$ is denoted by $|\Omega|$. \kotak
\end{definition}

In Definition \ref{def:searchengine}, we assume that $\Omega$ is made of a set of index of terms $t_{k_i}$, we will call it as a space of term. So, the web pages and queries are represented as vectors in $\Omega$ is also a space of event, whereby the semantics of this space is that of a multidimensional space. Therefore, a term $t_k$ is represented as a vector of web pages, i.e., the meaning of a term to be $\omega_k\in\Omega$ in which $t_k$ occurs. Let $q$ is a query, then $t_k\in q$, for $t_k = (w_1w_2\dots w_k)$. In logical implication, a web page is relevant to a query if it implies the query, that is if $\omega\Rightarrow q$ is true or $\omega\Rightarrow t_k$ is true $\forall\omega\in\Omega$: $(\omega\Rightarrow t_k) = 1$ \cite{nasution2012}, but for $2^k-2$ anothers of $\{\{t_k^{2^k-2}\} \subset \{w_1,w_2,\dots,w_k\} = t_k\}$, also $\omega\Rightarrow q$ is true $\forall \{t_k^{2^k-2}\} \ne\emptyset$. Thus, the degree of $\omega\Rightarrow q$ measured by $P(\omega\Rightarrow q)$, and probability $t_x$ in power subsets of $\{w_1,w_2,\dots,w_k\}$,
\begin{equation}
P(t_k) = \frac{1}{2^k-1}, t_k = (w_1w_2\dots w_k).
\end{equation}
Therefore there are an uniform mass probability function for $\Omega$,
\begin{equation}
P : \Omega\rightarrow [0,1]
\end{equation}
where $\sum_\Omega P(\omega) =  1$. 

\begin{definition}
\label{def:singleton}
Let $t_x$ is a search term, and $t_x\in{\cal S}$ where ${\cal S}$ is a set of singleton search term of search engine. A vector space $\Omega_x\subseteq\Omega$ is a singleton search engine event (\emph{singleton space of event}) of web pages that contain an occurrence of $t_x\in\omega_x$. The cardinality of $\Omega_x$ is denoted by $|\Omega_x|$. \kotak
\end{definition}

In the singleton space of event, $\Omega_x\subseteq\Omega$ if $\omega\Rightarrow t_x$ is true, or
\begin{equation}
\label{pers:logic}
\Omega_x(t_x) = \cases{1 & if $t_x$ is true at $\omega\in\Omega$,\cr
0 & otherwise\cr}
\end{equation}
and the cardinality of $\Omega_x$ be $|\Omega_x| = \sum_\Omega(\Omega_x(t_x)=1)$. This means that every web page that is indexed by search engine contains at least one occurrence of a search term, then we can measure its degree of uncertainty of $\omega\Rightarrow t_x$ on $\omega\Rightarrow q$ by  
\begin{equation}
P(\Omega_x) = P(\Omega_x(t_x)=1) = \frac{\sum_\Omega(\Omega_x(t_x)=1)}{|\Omega|} = \frac{|\Omega_x|}{|\Omega|}
\end{equation} 

For example, a search term is a person name: $x = $ "Mahyuddin Khairuddin Matyuso Nasution", then $\{t_x\} = \{w_1,w_2,w_3,w_4\} = \{$"Mahyuddin","Khairud\-din","Matyuso","Nasution"$\}$. At the time of doing the experiment, a Yahoo! search for "Mahyuddin Khairuddin Matyuso Nasution" returned $|\Omega_x| = ?$ hits or $|\Omega_x| = 3,440$ for "Mahyuddin K. M. Nasution", and the number of hits of search for $w_{i=1,2,3,4}$ are in $\{54,300;3,187,000;0;275,000\}$. The vector space of $t_x$ is of $\{t_k^{2^k-1}\} = 
\{\{w_1\},\{w_2\},\{w_3\},\{w_4\},\{w_1,w_2\},\{w_1,w_3\},\{w_1,w_4\},$ $\{w_2,w_3\},\{w_2,$ $w_4\},\{w_3,w_4\},~\{w_1,w_2,w_3\},\{w_1,w_2,w_4\},~\{w_2,w_3,w_4\},\{w_1,w_2,$ $w_3,w_4\}\}$. We have also $|\Omega_{x_p}| = 55$ from Yahoo search engine for $t_x$ with its pattern as a meaning core of $\Omega_x$,
\begin{equation}
\label{pers:singletonpattern}
|\Omega_{x_p}| = \sum_\Omega(\omega_x\Rightarrow t_x) \leq |\Omega_x|,
\end{equation}
where $\sum_\Omega(\omega_x\Rightarrow t_x)$ is the number of web pages containing $t_x$ with the pattern exactly. The singleton space of event captures in a particular sense all background knowledge about the search terms concerned available on the Web, geometrically this is a representation of meaning semantically.    

Similarly, for two search terms $t_x$ and $t_y$ in the different queries, we have 
\begin{equation}
\label{pers:independent}
\Omega_x\cap\Omega_y = (\Omega_x(t_y)=0)\wedge(\Omega_y(t_x)=0) = 
\emptyset
\end{equation}
i.e. any two singleton spaces of event are independent.
\begin{problem}
\label{prob:singleton}
Let $t_x$ and $t_y$ are two different search terms, $t_x\ne t_y$. Let $\Omega_x$ and $\Omega_y$ are the singleton search engine events of $t_x$ and $t_y$, respectively, and $|t_y|<|t_x|$ or $\forall w_i\in t_y$, $w_i\in t_x$, $\exists w_j\in t_x$, $w_j\not\in t_y$, then 
\begin{equation}
\label{pers:problem}
|\Omega_x| \stackrel{?}{=} |\Omega_x|+|\Omega_y|
\end{equation}
where $\Omega_x,\Omega_y\subseteq\Omega$.
\end{problem}

Problem \ref{prob:singleton} is a property of relation between any search engine and any query in a heterogeneous environment such as Web, and the information about any object to be scattered in various places. So in almost all measurements the bias exist.

\section{The Adaptive Properties in Search Engine}
Numerous studies of natural language processing (NLP) and Semantic Web utilize a search engine, mainly to obtain a set of documents that include a given query and to get statistical information about an object such as hit count of entity name, but to bring the NLP and Semantic Web to life such as the information processing services provide the knowledge, for example: ontology construction, knowledge extraction, question answering, and other purposes \cite{cimiano2004} needs more effort. 

Some properties we will derive to learn how to get the efficient ways to access and extract information from web. The purpose of this construction is to eliminate the bias by developing the adaptive model of relation between a search engine and the search terms. 

\begin{lemma}
\label{lem:singleton}
Let $t_x$ and $t_y$ are search term. If $t_x\ne t_y$, $t_x\cap t_y\ne\emptyset$ and $|t_y|<|t_x|$, then singleton search engine event of $t_x$ and $t_y$ is $\Omega_x = \Omega_x\cup\Omega_y$ or 
\begin{equation}
\label{pers:lemma1}
|\Omega_x| = |\Omega_x|+|\Omega_y|,
\end{equation} 
where $\Omega_x,\Omega_y\subseteq\Omega$.
\end{lemma}
\begin{proof}
For all search terms $t_x$ and $t_y$ where $t_x\ne t_y$, $t_x\cap t_y\ne\emptyset$ and $|t_y|<|t_x|$, by Definition \ref{def:term} and Definition \ref{def:searchengine} we have $\forall w_y\in t_y$, $w_y\in t_x$, $\exists w_x\in t_x$, $w_x\not\in t_y$ $\Rightarrow$ $\forall w_y\in\omega_y$, $w_y\in\omega_x$, $\exists w_x\in\omega_x$, $\omega_x\not\in\omega_y$ such that
\begin{equation}
\label{pers:lemsingleton1}
t_x\cap t_y = t_y {~\rm and~} t_x\cup t_y = t_x
\end{equation}
and
\begin{equation}
\label{pers:lemsingleton2}
\omega_x\cap\omega_y = \omega_y {~\rm and~} \omega_x\cup\omega_y = \omega_x.
\end{equation}
By Eq. (\ref{pers:independent}), clear that $\Omega_x\ne\Omega_y$ and $|\Omega_x\cap\Omega_y|=0$, then we have
\begin{equation}
\label{pers:lemsingleton3}
|\Omega_x\cup\Omega_y| = |\Omega_x|+|\Omega_y|.
\end{equation}
Let $\Omega_x = \{(t_x,\omega_x)\}$, based on meaning Eq. (\ref{pers:lemsingleton1}) and Eq. (\ref{pers:lemsingleton2}), we have $\Omega_x = \{(t_x,\omega_x)\} = \{(t_x\cup t_y,\omega_x\cup\omega_y)\} = \{(t_x,\omega_x)\cup(t_y,\omega_y)\} = \{(t_x,\omega_x)\}\cup\{(t_y,\omega_y)\} = \Omega_x\cup\Omega_y$. Therefore based on Eq. (\ref{pers:lemsingleton3}) the Eq. (\ref{pers:problem}) in Problem \ref{prob:singleton} be $|\Omega_x| = |\Omega_x|+|\Omega_y|$. \kotakt
\end{proof}

\begin{proposition}
\label{prop:recursive}
Let $t_z,\dots,t_y,t_x$ are search terms, where $t_z\ne\dots\ne t_y\ne t_x$ and $|t_z|<\dots<|t_y|<|t_x|$, then $\Omega_x = \Omega_x\cup\Omega_y$ holds recursively or $|\Omega_x| = |\Omega_x|+|\Omega_y|$, $\Omega_x,\Omega_y\subseteq\Omega$.
\end{proposition}
\begin{proof}
By the Lemma \ref{lem:singleton} and an assumption that $|t_z|<\dots<|t_y|<|t_x|$, we obtain $|t_z|<|t_{z_1}| \Rightarrow |\Omega_{z_1}| = |\Omega_{z_1}|+|\Omega_z|$, $|t_{z_1}|<|t_{z_2}| \Rightarrow |\Omega_{z_2}| = |\Omega_{z_2}|+|\Omega_{z_1}|$, $\dots$, $|t_y|<|t_x|\Rightarrow|\Omega_x|=|\Omega_x|+|\Omega_y|$.
Because of the inter-independence in the queries such as Eq. (\ref{pers:independent}), we obtain $\Omega_x\cap\Omega_y = \emptyset$, $\dots$, $\Omega_{z_1}\cap\Omega_z=\emptyset$, and $\Omega_x\cup\Omega_y$ belonging to $\Omega_x$, then 
\[
\begin{array}{rcl}
|\Omega_x| &=& |\Omega_x\cup\Omega_y|\cr
           &=& |\Omega_x|+|\Omega_y|\cr
           &=& |\Omega_x|+|\Omega_y\cup\dots|\cr
           &=& |\Omega_x|+|\Omega_y|+\dots\cr
           &=& |\Omega_x|+|\Omega_y|+|\dots\cup\Omega_z|\cr
           &=& |\Omega_x|+|\Omega_y|+\dots+|\Omega_z|\cr
\end{array}
\]
or $|\Omega_x| = |\Omega_x| + |\Omega_y|$ be recursive, where $|\Omega_y|+\dots+|\Omega_z|$ is a part of $|\Omega_x|$. \kotakt
\end{proof} 

\begin{lemma}
\label{lem:counter}
If $t_y\ne t_z$ and $t_y\cap t_z = \emptyset$, then $|\Omega_y\cap\Omega_z|=0$ and $|\Omega_y\cup\Omega_z| = |\Omega_y|+|\Omega_z|$.
\end{lemma}
\begin{proof}
For all search terms $t_y$ and $t_z$ where $t_y\ne t_z$ and $t_y\cap t_z =\emptyset$, by Definition \ref{def:term} and Definition \ref{def:searchengine} we obtain $\forall w_y\in t_y$, $w_y\not\in t_z$ $\wedge$ $\forall w_z\in t_z$, $w_z\not\in t_y$ $\Rightarrow$ $\forall w_y\in\omega_y$, $w_y\not\in\omega_z$ $\wedge$ $\forall w_z\in\omega_z$, $w_z\not\in\omega_y$ such that
\begin{equation}
\label{pers:counter1}
t_z\cap t_y=\emptyset \vee t_z\cup t_y = t_y\cup t_z
\end{equation}
and
\begin{equation}
\label{pers:counter2}
\omega_y\cap\omega_z = \emptyset \vee \omega_y\cup\omega_z=\omega_z\cup\omega_y
\end{equation}
Let $\Omega_y = \{(t_y,\omega_y)\}$ and $\Omega_z = \{(t_z,\omega_z)\}$ are two independent events from the queries, based on Eq. (\ref{pers:independent}) we obtain $\Omega_y\cap\Omega_z = \emptyset$ and 
\begin{equation}
\label{pers:counter3}
|\Omega_y\cap\Omega_z|=0
\end{equation}
and by combining the meaning of (\ref{pers:counter1}), (\ref{pers:counter2}), (\ref{pers:counter3}), and $\{(t_y,\omega_y)\}\cup\{(t_z,\omega_z)\} = \Omega_y\cup\Omega_z$ and we can conclude that $|\Omega_y\cup\Omega_z| = |\Omega_y|+|\Omega_z|$. \kotakt
\end{proof}

Lemma \ref{lem:counter} expresses that Eq. (\ref{pers:problem}) in Problem \ref{prob:singleton} be $|\Omega_x| \ne |\Omega_x|+|\Omega_y|$ or $|\Omega_x\cup\Omega_y| = |\Omega_x|+|\Omega_y|$.

\begin{proposition}
\label{prop:kunci}
Let $\Omega_x\cap\Omega_y=\emptyset$ and $\Omega_a\cap\Omega_b=\emptyset$. If $|\Omega_x|=|\Omega_x|+|\Omega_a|$ and $|\Omega_y|=|\Omega_y|+|\Omega_b|$, then $|\Omega_x\cap\Omega_y|\geq 0$.
\end{proposition}
\begin{proof}
This is a direct consequence of Lemma \ref{lem:singleton} and Lemma \ref{lem:counter}. \kotakt
\end{proof}

\begin{lemma}
\label{lem:sama}
Let $t_x$ and $t_z$ are search terms. If $t_x\ne t_z$, $t_x\cap t_z=\emptyset$, and $\omega_x\cap\omega_z\ne\emptyset$, then $|\Omega_x|=|\Omega_z|$, $\Omega_x,\Omega_z\subseteq\Omega$.
\end{lemma}
\begin{proof}
For all search terms $t_x$ and $t_z$ where $t_x\ne t_z$, $t_x\cap t_z=\emptyset$ and $\omega_x\cap\omega_z\ne\emptyset$, by Definition \ref{def:term} and Definition \ref{def:searchengine} we obtain $\forall w_x\in t_x$, $w_x\not\in t_z$, and $\forall w_z\in t_z$, $w_z\not\in t_x$ then
\begin{equation}
\label{pers:sama1}
t_x\cap t_z = \emptyset \vee t_x\cup t_z=t_z\cup t_x,
\end{equation}
but $\forall w_x\in\omega_x$, $w_x\in\omega_z$ and $\forall w_z\in\omega_z$, $w_z\in\omega_z$ then
\begin{equation}
\label{pers:sama2}
\omega_x\cap\omega_z = \omega_x = \omega_z, \omega_x\cup\omega_z = \omega_z\cup\omega_x = \omega_x = \omega_z.
\end{equation}
For $\Omega_x = \{(t_x,\omega_x)\}$ and $\Omega_z = \{(t_z,\omega_z)\}$ we have $\Omega_x\cap\Omega_z = \{(t_x,\omega_x)\}\cap\{(t_z,\omega_z)\} = \{(t_x,\omega_z)\}\cap\{(t_z,\omega_z)\}$, and because $t_z\in\omega_z$ the intersection of $t_x\cap t_z$ must be $\{(t_x,\omega_z)\}\cap\{(t_z,\omega_z)\} = \{(t_z,\omega_z)\}\cap\{(t_z,\omega_z)\}$ or $\Omega_x\cap\Omega_z = \Omega_z\cap\Omega_z$ or $\Omega_x\cap\Omega_z=\Omega_z$. Similarly, $\Omega_x\cap\Omega_z=\Omega_x$. Thus $|\Omega_x|=|\Omega_z|$. \kotakt
\end{proof}

This Lemma explains that Eq. (\ref{pers:problem}) in Problem \ref{prob:singleton} be $|\Omega_x|=|\Omega_y|$ if and only if $t_x\ne t_y$ but $t_x,t_y\in\omega_x$ $\wedge$ $t_x,t_y\in\omega_y$. In other word, based on combining (\ref{pers:sama1}) and (\ref{pers:sama2}) $\Omega_x = \{(t_x,\omega_x)\} = \{(t_x,\omega_x\cup\omega_y)\} = \{(t_x,\omega_x)\cup(t_x,\omega_y)\} = \{(t_y,\omega_x)\cup(t_y,\omega_y)\} = \{(t_y,\omega_x\cup\omega_y)\}= \{(t_y,\omega_y)\} = \Omega_y$. This shows that the search terms may be different but they come from same web pages, and in this case they take the same meaning from web. 

\section{Conclusions and Future Work}
Studying to properties of relation between query and search engine gave the understanding about 
the semantic representation statistically for object in literal text. Our near future work is to generate some properties of search engine for doubleton.

\end{document}